\documentclass[reqno, letterpaper]{article}

\usepackage{amsmath}
\usepackage{amsfonts}
\usepackage{amssymb}
\usepackage{amsthm}

\newcommand{\oemargin}[1]{\oddsidemargin#1 \evensidemargin#1}
\oemargin{0.3in}

\DeclareMathOperator*\essinf{essinf}

\newcommand{\abs}[1]{\left| #1 \right|} 
\newcommand{\set}[1]{\left\{#1\right\}} 
\newcommand{\sets}[2]{\set{#1\,:\,#2}} 
\newcommand{\inds}[1]{ {\mathbf 1}_{\set{#1}}} 
\newcommand{\ind}[1]{ {\mathbf 1}_{{#1}}} 
\newcommand{\cd}{\text{c\` adl\` ag} } 
\newcommand{\norm}[1]{{||#1||}} 

\newcommand{\prf}[1]{ ( #1 )_{t\in [0,T]}}

\newcommand{\RN}[2]{\frac{d#1}{d#2}}

\newcommand{\dqp}{\RN{\QQ}{\PP}}

\providecommand{\R}{} \renewcommand{\R}{{\mathbb R}}
 
\newcommand{\N}{{\mathbb N}}
\newcommand{\PP}{{\mathbb P}}
\newcommand{\QQ}{{\mathbb Q}}

\newcommand{\EE}{{\mathbb E}}
\newcommand{\FF}{{\mathcal F}}

\newcommand{\BB}{{\mathcal B}}

\newcommand{\UU}{\mathbb{U}}

\newcommand{\MM}{{\mathcal M}}

\newcommand{\EN}{{\mathcal E}}
\newcommand{\XX}{{\mathcal X}}
\newcommand{\YY}{{\mathcal Y}}

\newcommand{\eps}{\varepsilon}
\newcommand{\ld}{\lambda}

\newcommand{\el}{{\mathbb L}} 
\newcommand{\lzer}{\el^0}
\newcommand{\lone}{\el^1}
\newcommand{\ltwo}{\el^2}

\numberwithin{equation}{section}
\theoremstyle{plain}                
\newtheorem{theorem}{Theorem}[section]
\newtheorem{lemma}[theorem]{Lemma}
\newtheorem{proposition}[theorem]{Proposition}
\newtheorem{corollary}[theorem]{Corollary}
\theoremstyle{definition}           
\newtheorem{definition}[theorem]{Definition}
\newtheorem{example}[theorem]{Example}
\newtheorem{problem}{Problem}[section]

\theoremstyle{remark}
\newtheorem{remark}[theorem]{Remark}

\renewcommand{\set}[1]{\{#1\}}
\newcommand{\ab}[1]{\langle #1 \rangle}
\newcommand{\Sl}{S^{\ld}}
\newcommand{\Sln}{S^{\ld^n}}
\newcommand{\Zl}{Z^{\ld}}
\newcommand{\Zln}{Z^{\ld^n}}

\newcommand{\LM}{\Lambda_M}
\newcommand{\MMl}{\MM^{\ld}}

\newcommand{\Fi}{{\mathbf F}}

\newcommand{\XXl}{\XX^{\ld}}
\newcommand{\YYl}{\YY^{\ld}}

\newcommand{\Ql}{\QQ^{\ld}}
\newcommand{\hX}{\hat{X}}
\newcommand{\hY}{\hat{Y}}

\newcommand{\hXl}{\hX^{\ld}}

\newcommand{\hXxl}{\hX^{x,\ld}}
\newcommand{\hYyl}{\hY^{y,\ld}}
\newcommand{\hYyln}{\hY^{y_n,\ld^n}}
\newcommand{\hXln}{\hX^{\ld^n}}

\newcommand{\hXxln}{\hX^{x,\ld^n}}

\newcommand{\vl}{v^{\ld}}
\newcommand{\vln}{v^{\ld^n}}
\newcommand{\vlz}{v^{\ld^0}}

\newcommand{\ul}{u^{\ld}}
\newcommand{\uln}{u^{\ld^n}}
\newcommand{\ulz}{u^{\ld^0}}

\newcommand{\Lyl}{L^{y,\ld}}
\newcommand{\Lyln}{L^{y_n,\ld^n}}
\newcommand{\Leb}{{\mathrm{Leb}}}
\newcommand{\hYylk}{\hat{Y}^{y_k,\ld^k}}

\newcommand{\seeq}[1]{\{#1\}_{n\in\N}}

\begin{document}

\begin{center}
{\LARGE\bf Stability of utility-maximization in incomplete markets}
\end{center}

\bigskip

\noindent\begin{minipage}{0.50\textwidth}
\begin{center}
Kasper Larsen\\ Department of Mathematical Sciences, \\
  Carnegie Mellon University,\\ Pittsburgh, PA 15213, \\ {\tt
    kasperl@andrew.cmu.edu}
\end{center}
\end{minipage}
\hfill
\begin{minipage}{0.5\textwidth}
\begin{center}
Gordan {\v Z}itkovi{\' c}\\
Department of Mathematics, \\
University of  Texas at Austin, \\
Austin, TX 78712, \\
{\tt gordanz@math.utexas.edu}
\end{center}
\end{minipage}

\bigskip

\begin{center}
 \today
\end{center}

\medskip

\begin{abstract}
  The effectiveness of utility-maximization techniques for portfolio
  management relies on our ability to estimate correctly the
  parameters of the dynamics of the underlying financial assets. In
  the setting of complete or incomplete financial markets, we
  investigate whether small perturbations of the market coefficient
  processes lead to small changes in the agent's optimal behavior derived from
  the solution of the related utility-maximization problems.
  Specifically, we identify the topologies on the parameter process
  space and the solution space under which utility-maximization is
  a continuous operation, and we provide a counterexample showing that
  our results are best possible, in a certain sense.  A novel result
  about the structure of the solution of the utility-maximization
  problem where prices are modeled by continuous semimartingales is
  established as an offshoot of the proof of our central theorem.
\end{abstract}

\vspace{+20pt}


\section{Introduction}
\paragraph{The Central Problem.} Financial theory in general, and
mathematical finance in particular, aim to describe and understand
the behavior of rational agents faced with an uncertain evolution
of asset prices. In the simplest, yet most widespread models of
such behavior, the agent has a fixed and immutable assessment of
various probabilities related to the future evolution of the
prices in the financial market. Taking her views as correct, the
agent proceeds to implement a dynamic trading strategy which is
chosen so as to maximize a certain nonlinear functional of the
terminal wealth - the utility functional. Often, the utility
functional is of the ``expected-utility'' type, i.e., the agent's
objective is to maximize $\UU(X_T)=\EE[U(X_T)]$ over all possible
random variables $X_T$ she can generate through various investment
strategies on a trading horizon $[0,T]$, starting from a given
initial wealth $x$.  $U(\cdot)$ is generally a concave and
strictly increasing real-valued function defined on the positive
semi-axis $(0,\infty)$, and is used as a  model of the agent's
risk preferences. In order to implement this program in practice,
the agent chooses a particular model of the evolution of asset
prices, estimates its parameters using the available market data,
and combines the obtained market specification with the particular
idiosyncratic form of the utility functional $\UU$. Having seen
how the choice of the market model requires imperfect measurement
and estimation, the natural question to ask is then the following:
\begin{quote}\em
``How are the agent's behavior and its optimality affected by
(small) misspecifications of the underlying market model?''
\end{quote}
Unless we can answer this question by a decisive {\em ``Not
much!''}, the utility-maximization framework as described above
loses its practical applicability.

In the classical setting of the theory of partial differential
equations, and applied mathematics in general, similar questions
have been posed early in the literature. It is by now a classical
methodological requirement to study the following three aspects of
every new problem one encounters:
\begin{enumerate}
\item existence,
\item uniqueness,
\item \label{ite:three}
sensitivity of the solution with respect to changes of the
  problem's input parameters.
\end{enumerate}
These criteria are generally known as {\em Hadamard's well-posedness
  requirements} (see \cite{Had02}). The present paper adopts the view
that the market model specification is one of the most important
input data in the utility-maximization problem, and focuses on the
third requirement with that in mind.

\paragraph{Existing research.} In the general setting of the
semimartingale stock-price model, the first two of the Hadamard's
requirements (existence and uniqueness) have been settled
completely by a long line of research reaching at least to Robert
Merton and continuing with the work of Chuang, Cox, He, Karatzas,
Kramkov, Lehoczky, Pearson, Pliska, Schachermayer, Shreve, Xu,
etc.\ (see \cite{Mer71}, \cite{Pli86}, \cite{CoxHua89},
\cite{HePea91}, \cite{KarLehShrXu91}, \cite{KraSch99}, merely to
scratch the surface). Tight conditions are now known on
practically all aspects of the problem which guarantee existence
and uniqueness of the optimal investment strategy. The question of
sensitivity has been studied to a much lesser degree and, compared
to the model-specification issues, much more effort has been
devoted to the perturbations of the shape of the utility function
or the initial wealth (see, e.g., \cite{JouNap04} and
\cite{CarRas05}). Related questions of  stability of option
pricing (under market perturbations) have been studied by
\cite{ElkJeaShr98}, for the case of the Samuelson's (also know as
Black-Scholes-Merton) market, and several authors have studied the
phase transition ``from discrete- to continuous-time models'',
see e.g., \cite{HubSch98} and the monograph \cite{Pri03}.

The concept of robust portfolio optimization, which has been studied
extensively in the financial and mathematical literature, is related
to our notion of stability.  The main goal of robust portfolio
optimization is to create decision rules that work well - at least up
to some degree - under each of several model specifications, or under
several probability measures (sets of beliefs) ${\mathcal
  Q}\in\mathcal{P}$ where $\mathcal{P}$ is a family of financial
models. A popular way of approaching this problem consists of allowing
for multiple model specifications, and considering investors who care
about expected utility, but in a different way in each of the possible
models. The starting point for this approach is the celebrated paper
\cite{GilSch89}, where the authors show how to relax the classical von
Neumann-Morgenstern preference axioms by introducing $X\mapsto
\inf_{{\mathcal Q}\in\mathcal{P}}\Big( \EE^{{\mathcal
    Q}}[U(X)]+\varrho({\mathcal Q})\Big)$ as the numerical
representation for the robust utility functional (see also
\cite{MMR04}). Here $X$ typically represents the terminal value of some
admissible trading strategy, and $\varrho$ assigns penalization
weights to the different possible model specifications ${\mathcal
  Q}\in\mathcal{P}$. We cannot give a complete overview of this theory
and its many aspects (one interesting property is how model ambiguity
interacts with the coefficient of risk aversion, see e.g.,
\cite{Troj2002}), but refer the reader to the textbook
\cite{FolSch02a} and the references therein.  We emphasize though,
that, while superficially similar to the robust optimization approach,
our analysis is based on the assumption that our investor firmly
believes that the original probability measure $\PP$ is correctly
specified, and does not incorporate any model ambiguity into her
optimal decision. If we view the perturbations of the model as the
perturbations of the underlying probability measure $\PP$ (via
Girsanov's theorem), one of the facets of our question of stability
can be reformulated as follows: Is the $\PP$-optimal strategy
approximately optimal for all elements in some small-enough set of
``nearby'' models ${\mathcal Q}\in{\mathcal P}$?  In other words, our
problem deals with the evaluation of the optimality properties of one
prespecified strategy in various market models, while the robust
optimization seeks a strategy with good properties under different
market models.

\paragraph{Our results.} In the present paper we investigate the
stability properties of utility-maximization in a wide class of
{\em complete or incomplete} financial models.  Specifically, we
develop a methodology which can deal with any financial market
with continuous asset prices, without restrictions on the
underlying filtration. In the setting of such models (described in
detail below, and including Samuelson's model as well as
stochastic volatility models) the concept of the
market-price-of-risk can be defined in an unambiguous way.
Moreover, one of our main technical results states that in these
models the maximal dual elements (in the sense of \cite{KraSch99})
are {\em local martingales} and admit a multiplicative
decomposition into a ``minimal local martingale density'', and an
``orthogonal part''. As a consequence, we show that in the setting
of the dual approach to utility-maximization, the dual optimizer
is always a local martingale when the stock price is continuous.
This extends a similar result from \cite{KarZit03} stated in the
more restrictive milieu of It\^o-process models.

When the model under scrutiny allows for a notion of volatility,
the market-price-of-risk can be interpreted as the drift, weighted
by a negative power of the volatility. In particular,
misspecifications of the market-price-of-risk translate into
homothetic misspecifications in the drift process. \cite{Rog01}
discusses the practical difficulties related to estimating the
drift and points out that the magnitude of the error attached to
the drift estimate is significant.  The continuity of the value
function, as well as the optimal terminal wealth of a
utility-maximizing agent - seen as functions of the
market-price-of-risk - constitute the center of our attention.
Therefore, our analysis is to be seen as stability with respect to
small drift misspecifications and hopefully provide some insight
also into the more complicated problem of large misspecifications
that \cite{Rog01} points at.

The value function of our utility-maximization problem 
takes values in the Euclidean space $\R$ and
there is little discussion about the proper notion of continuity
there. However, the market-price-of-risk (in the domain), and the
optimal terminal wealths (in the co-domain), are more complicated
objects (a stochastic process and a random variable), and present
us with a variety of choices for the topology under which the
notion of ``perturbation'' can be interpreted. One of the
contributions of this paper is to identify a class of topologies
on the domain, and a particular topology (of convergence in
probability) on the co-domain, under which utility-maximization
becomes a continuous operation when a simple condition of
$V$-relative compactness is satisfied. Under the additional assumption
that all the markets under consideration are complete, we show that
$V$-relative compactness is, in fact, both necessary and sufficient.
Moreover, we provide an example, set
in a complete It\^ o-process financial market, in which a very
strong convergence requirement imposed on the market-price-of-risk
processes still fails to lead to any kind of convergence of the
corresponding optimal terminal wealths. 

On the technical side, the proof of our main stability result
requires an analysis of the structure of the solution of the
utility-maximization problem. Specifically, a recourse to
convex-duality techniques is of great importance; most of the
intermediate steps leading to the final result deal with the dual
optimization problem and its properties, and  for every continuity
result in the primal problem, there is a corresponding continuity
result in the dual. It is in the heart of the duality approach in
convex optimization that one can choose whether to work on the
primal or the dual problem - depending on which one is more
amenable to analysis in a particular situation - and easily
translate the obtained results to the other one. In our case, the
advantage of the dual problem is that certain close substitutes
for  {\em compactness} (such as the use of Komlos' lemma) bring a
number of topological techniques into play. One of the
mathematical messages of this paper is that the use of duality
theory is not restricted to the existence results only, but can be
put to a more versatile use.

The structure of the paper follows a simple template: The next
section describes the modeling framework, poses the problem and
states the main results. Section 3 invokes some important facts
about the convex-duality treatment of utility-maximization
problems and provides a proof of the main result through a sequence of lemmas.
The Appendix contains an auxiliary result exemplifying the notion
of appropriate topology.

\section{The Problem Formulation and the Main Results}
\subsection{The model framework}
Let $(\Omega,\FF,\PP)$ be a complete probability space, and let
$\Fi=\prf{\FF_t}$, be a filtration satisfying the usual conditions.
 For  a continuous $\Fi$-local martingale $M=\prf{M_t}$,  let
$\Lambda$ denote the set of all predictable processes
$\ld=\prf{\ld_t}$ with the property that \[\int_0^T \ld_u^2\,
d\ab{M}_u<\infty,\ \text{a.s.},\] where, as usual,
$\ab{M}=\prf{\ab{M}_t}$ denotes the quadratic variation of the
local martingale $M$. Each $\ld\in\Lambda$ defines a continuous
semimartingale $\Sl$, where
\begin{equation}
   \label{equ:models}
   \begin{split}
\Sl_t=1+M_t+\int_0^t \ld_u\, d\ab{M}_u,\ t\in [0,T].
   \end{split}
\end{equation}
Together with the trivial bond-price process $B_t\equiv 1$, $\Sl$
constitutes
a financial market. In the sequel, we will simply write {\em the market
  $\Sl$}.
\begin{example}
The proto-example for the family $\sets{\Sl}{\ld\in\Lambda}$ is the
class of It\^o-process markets of the form
\begin{equation}
   \nonumber
   \begin{split}
d\Sl_t=\Sl_t\big( \mu^{\ld}_t\, dt+\sigma_t dB_t\big),\ \Sl_0=1, \text{where }
\mu^{\ld}_t=
\ld_t \sigma^2_t,
   \end{split}
\end{equation}
defined on the filtration $\Fi=\prf{\FF_t}$, generated either by a
Brownian motion $B$, or by a pair $(B,W)$ of independent Brownian
motions. In the first case the market is complete, but in the
second case the market is incomplete. The important continuous
models of financial markets such as the Samuelson's model, or the
class of stochastic volatility models, fall within this framework.
\end{example}
\begin{remark} The choice of the special form for  the model class in
\eqref{equ:models} is not arbitrary. In fact, it is a consequence
of the main result of \cite{DelSch95}, that {\em any}
continuous arbitrage-free (num\'eraire-denominated) model of a stochastic
market admits the representation \eqref{equ:models}.
\end{remark}
\subsection{Absence of arbitrage and its consequences}
For $\lambda \in \Lambda$ the stochastic exponential process
$\Zl=\prf{\Zl_t}$, given by
\begin{equation}
   \nonumber
   \begin{split}
\Zl_t=\EN(-\ld\cdot M)_t=\exp\left(-\int_0^{t} \ld_u\,
dM_u-\frac{1}{2} \int_0^{t} \ld_u^2\, d\ab{M}_u\right),\ t\in
[0,T].
   \end{split}
\end{equation}
is a strictly positive local martingale and acts as a
state-price-deflator for $\Sl$. More precisely, It\^o's formula
implies that the process $\Zl X$ is a local martingale for each
semimartingale $X$ of the form $X=H\cdot \Sl$, whenever $H$ is a
predictable and $\Sl$-integrable (i.e., $H\in L(\Sl)$). When $\Zl$
is a genuine martingale, the measure $\Ql\sim\PP$ defined by
\begin{equation}
   \nonumber
   \begin{split}
\RN{\Ql}{\PP}=\Zl_T
   \end{split}
\end{equation}
is a probability measure under which the stock-price process $\Sl$
is a local martingale. In that case, the market $\Sl$ satisfies
the condition of No Free Lunch with Vanishing Risk (NFLVR). It is
customary to call $\Ql$ the {\em minimal local martingale
measure}. In general, the set $\MMl$ of equivalent local
martingale measures (i.e., all probability measures $\QQ$,
equivalent to $\PP$, under which the process $\Sl$ is a local
martingale) is larger than just a singleton. The following result
is a direct consequence of Theorem 1 in \cite{Sch95a}, which, in
turn, is a generalization of the results in \cite{AnsStr92} and
\cite{AnsStr93}.
\begin{proposition}[Schweizer, Ansel, Stricker]
  \label{pro:Schweizer-structure-of-M}
When $\MMl\not=\emptyset$, every probability
measure $\QQ\in\MMl$ has the  form
\begin{equation}
   \nonumber
   \begin{split}
\frac{d\QQ}{d\PP}=\Zl_T \EN(L)_T,
   \end{split}
\end{equation}
for some local martingale $L$ strongly orthogonal to $M$, meaning
$\langle L,M\rangle \equiv 0$.
\end{proposition}
It is an unexpected result of \cite{DelSch98a} that the market
$\Sl$ can satisfy NFLVR, without the density process
$\Zl$ having  the martingale property. In that case, the minimal
martingale measure does not exist.

We do not postulate that the process $\Zl$ is a (uniformly
integrable) martingale. Instead, we restrict our attention to the
set $\LM\subseteq \Lambda$, containing all $\ld\in\Lambda$ such
that the financial market $\Sl$ admits NFLVR.  The existence of an
equivalent martingale measure for the process $\Sl$, $\ld\in\LM$,
now follows from the celebrated Fundamental Theorem of Asset
Pricing of Delbaen and Schachermayer (\cite{DelSch94}).
\begin{remark}
Even though we will only consider $\ld\in\LM$, all the results in
the sequel can be extended to the most general case
$\ld\in\Lambda$. Admittedly, in this general case, the markets
under consideration will not be arbitrage free in the sense of
NFLVR, but the existence of a strictly positive
state-price-deflator $\Zl$ turns out to be enough. This (mild)
generalization would add to the technicalities of the proofs
without adding much to the content, so we have choosen not to
pursue it.
\end{remark}

\subsection{The utility-maximization problem}
\begin{definition}
\label{def:utility}
A strictly
concave, strictly increasing $C^1$-function $U:(0,\infty)\to\R$ satisfying
the Inada conditions:
\begin{equation}
   \nonumber
   \begin{split}
\lim_{x\to 0} U'(x)=+\infty,\ \lim_{x\to\infty} U'(x)=0,
   \end{split}
\end{equation}
as well as the reasonable asymptotic elasticity condition $AE[U]<1$, where
\begin{equation}
   \nonumber
   \begin{split}
AE[U]=\begin{cases} \limsup_{x\to\infty}\frac{xU'(x)}{U(x)},&
\lim_{x\to\infty} U(x)=+\infty,\\ 0,& \text{otherwise.}
\end{cases}
   \end{split}
\end{equation}
is called a {\em reasonably elastic utility function}.
\end{definition}
\begin{remark}
As usual, we extend the utility function $U$ to the negative
semi-axis by defining $U(x)=-\infty$ for negative $x$-values.
\end{remark}
Given a financial market $\Sl$ with $\ld\in\Lambda$, the {\em
utility-maximization problem} for a financial agent with initial
wealth $x>0$ (and the risk attitude described by the utility
function $U$) is to maximize the expected utility $\EE[U(X_T)]$ over
all terminal values of the wealth processes obtainable by trading in
the stock $\Sl$ and investing in the risk-free security in a
self-financing manner. More precisely, the utility-maximization
problem is posed through its value function $\ul:\R_+\to\R$, where
\begin{equation}
   \label{equ:UMP}
   \begin{split}
\ul(x)=\sup_{X\in\XXl(x)}\EE[U(X_T)].
   \end{split}
\end{equation}
and $\XXl(x)$ is the usual class of wealth processes constrained by
an admissibility requirement in order to rule out the doubling
strategies
\begin{equation}
   \nonumber
   \begin{split}
\XXl(x)=\sets{ x+H\cdot \Sl}{H\in L(\Sl),\ x+H\cdot \Sl \text{ is a non-negative process}}.   \end{split}
\end{equation}
A number of authors have studied the problem \eqref{equ:UMP} on
various levels of generality. Culminating with \cite{KraSch99},
this line of research has established a natural set of regularity
assumptions on the market and on the utility function, under which
\eqref{equ:UMP} admits a unique solution $\prf{\hXxl_t}\in
\XXl(x)$, and the value function $x\mapsto \ul(x)=\EE[U(\hXxl_T)]$
is finite-valued and continuously differentiable.
\subsection{The central problem}
Now that we have introduced all the needed elements, we can pose
our stability problem for the utility-maximization problem
\begin{problem} Given an initial wealth  $x>0$,
let the sequence $\{\lambda^n\}_{n\in\mathbb{N}}$ in $\Lambda_M$
converging to $\ld^0\in\Lambda_M$ is some topology.
 Under which conditions on
the sequence $\{\lambda^n\}_{n\in\mathbb{N}}$ and the topology in
which it converges to $\ld^0$, will
\begin{enumerate}
\item the value functions $\uln(x)$,
\item the optimal terminal wealths $\hXxln_T$,
\end{enumerate}
converge to the corresponding value function $u^{\lambda^0}(x)$
and the corresponding optimal terminal wealth
$\hat{X}^{x,\lambda^0}_T$?
\end{problem}
\subsection{Appropriate topologies}
Before we give a precise statement of our main result, we
introduce and comment on a class of topologies in the present
subsection, as well as the concept of $V$-relative compactness, in
the following subsection. Standardly, $\lzer$ denotes the set of
all (equivalence classes) of $\FF$-measurable finite-valued random
variables and $\lzer_+$ denotes its positive cone.
\begin{definition}
\label{def:appropriate} A metrizable topology $\tau$ on $\Lambda$
is said to be {\em appropriate} if the mapping  $\ld\mapsto \Zl_T$
of $\Lambda$ into $\lzer_+$ is continuous, when $\Lambda$ is
endowed with $\tau$, and $\lzer$ with the topology of convergence
in probability.
\end{definition}

\begin{remark}
The requirement of metrizability in the Definition
\ref{def:appropriate} is imposed only to simplify the analysis
below as it allows us to circumvent the use of nets. Any topology
for which $\ld\mapsto\Zl_T$ is continuous can be weakened to a
metrizable topology with the same property.
\end{remark}
The following example describes two natural appropriate topologies.
\begin{example}
\label{exa:appropriate}
\
\begin{enumerate}
\item Let the positive measure $\mu^M$, defined on predictable
  $\sigma$-algebra on the product space
$[0,T]\times\Omega$,
be given by
\begin{equation}%
\label{equ:mu-A}
    \begin{split}
\mu^M(A)=\EE\int_0^T \ind{A}(t) \,d\ab{M}_t.
    \end{split}
\end{equation}
Proposition \ref{pro:conv-Fatou} in the Appendix
states that the restriction of the $\ltwo(\mu^M)$-norm
\begin{equation}
   \nonumber
   \begin{split}
\norm{\ld}_{\ltwo{(\mu^M)}}^2=\EE\int_0^T \ld_u^2\, d\ab{M}_u,
   \end{split}
\end{equation}
onto $\sets{\ld\in\Lambda}{ \norm{\ld}_{\ltwo(\mu^M)} <\infty}$
induces an appropriate topology.
\item Another example of an appropriate topology is the so-called \emph{ucp}-topology (uniform
convergence on compact sets in probability), when restricted to
left-continuous processes in
  $\LM$. In other words, a sequence $\seeq{\ld^n}$ converges to $\ld$ in ucp if the
  sequence
\[ \sup_{t\in [0,T]} \abs{\ld^n_t-\ld_t} \] of
random variables converges to $0$ in probability. For more
information about the ucp-toplogy, see Section II.4 in
\cite{Pro04}.
\end{enumerate}
\end{example}

\subsection{The $\log$-example}

To the reader in acquiring a  better understanding of our main result,
we provide a simple example that
illustrates the use of appropriate topologies in  the
$\mathbb{L}^2(\mu^M)$-class, i.e.,  the  sequence of models with
square integrable market-price-of-risk processes, $\lambda^n\in
\mathbb{L}^2(\mu^M)$ for $n\in\N$. We consider an investor with
$U(x)=\log(x)$ (so-called {\em $log$-investor}). It is well-known that
that her behavior is myopic, in the sense that the 
optimal wealth is given by 
\begin{align}\label{eq:GOP}
 \hat{X}^{x,\lambda^n}_T = \frac{x}{Z^{\lambda^n}_T}.
\end{align}
Thanks to Proposition \ref{pro:conv-Fatou}, 
if $\lambda^n \to \lambda^0$ in $\mathbb{L}^2(\mu^M)$,
then $Z^{\lambda^n}_T \to Z^{\lambda^0}_T$. Consequently, 
 for the optimal wealths, given by \eqref{eq:GOP}, we have 
$\hat{X}^{x,\lambda^n}_T \to \hat{X}^{x,\lambda^0}_T$ in probability. 
Furthermore, inserting \eqref{eq:GOP} into
\eqref{equ:UMP} yields the following expression for the value
function
\begin{align*} u^n(x) =
\EE\left[\log\left(\hat{X}^{x,\lambda^n}_T\right)\right] =
\log(x)+
\EE\left[\int_0^T\lambda^n_udM_u+\frac12\int_0^T\left(\lambda^n_u\right)^2d\langle
M\rangle_u\right].
\end{align*}
Since  $\lambda^n \in \mathbb{L}^2(\mu^M)$ for all $n\in\N$, 
the
 stochastic integral in the expression above is a 
genuine martingale and 
the following representation holds
$$
u^n(x)=\log(x)+\frac12 \| \lambda^n\|^2_{\mathbb{L}^2(\mu^M)}.
$$
This relation shows that the requirement $\lambda^n \in
\mathbb{L}^2(\mu^M)$ grants finiteness of the value function $u^n$. It
also implies that the convergence $\lambda^n\to\lambda^0$ in
$\mathbb{L}^2(\mu^M)$ implies pointwise convergence of the value
functions $u^n(\cdot)$ to $u^0(\cdot)$.

For an investor with a general utility function $U(\cdot)$, the
corresponding optimizer $\hat{X}^{x,\lambda}_T$ can be a lot more
complicated than \eqref{eq:GOP}, and, as we illustrate, more
regularity needs to be imposed in order to obtain positive
results. This is the content of the next subsection.

\subsection{$V$-relative compactness}
A reasonably elastic  utility function $U$ (as in Definition
\ref{def:utility}) is linked via conjugacy to its {\em
Legendre-Fenchel transform} $V:(0,\infty)\to\R$ given by
\begin{equation}
   \nonumber
   \begin{split}
V(y)=\sup_{x>0} \big( U(x)-xy\big).
   \end{split}
\end{equation}
\begin{definition}
\label{def:V-relative-weak-compactness}
A subset $\Lambda'$ of $\Lambda$ is said to be {\em $V$-relatively
  compact} if the following family of random variables
\begin{equation}
   \label{equ:V-of-Z}
   \begin{split}
     \sets{V(\Zl_T)}{\ld\in\Lambda'}
   \end{split}
\end{equation}
is uniformly integrable.
\end{definition}
\begin{remark}
It is enough to replace $V(\Zl_T)$ by
  $V^+(\Zl_T)=\max(V(\Zl_T),0)$ in \eqref{equ:V-of-Z}.  Indeed,
  the family $\sets{\Zl_T}{\ld\in\Lambda}$ is contained in the unit ball of $\lone$, and
  concavity properties of the
  function $V^-(\cdot)=\max(0,-V(\cdot))$ can be used to conclude that
  $\sets{V^-(\Zl_T)}{\ld\in\Lambda}$ is uniformly integrable (see the first part
  of the proof of Lemma 3.2, p.~914 in \cite{KraSch99} for more details).
\end{remark}

\subsection{The main result}
\begin{theorem}
 \label{thm:main}
 Let $\Lambda'$ be a $V$-relatively compact subset of $\Lambda_M$, and
 let $\tau$ be an appropriate topology. Then for any
 $\ld\in\Lambda'$, the function $\ul:(0,\infty)\to\R$ is
 finite-valued, and for each $x>0$ there exists an a.s.-unique optimal
 terminal wealth $\hXxl_T$ (the last element of the wealth process
$\hXxl\in\XXl(x)$), for the utility
 maximization problem \eqref{equ:UMP}. Moreover, the mappings
\begin{equation}
   \nonumber
   \begin{split}
\Lambda' \times (0,\infty) \ni (\ld,x) & \mapsto \ul(x)\in\R, \text{
  and } \\
\Lambda' \times (0,\infty) \ni (\ld,x) & \mapsto \hXxl_T \in \lzer_+
   \end{split}
\end{equation}
are jointly continuous when $\Lambda'$ is equipped with $\tau$, and
$\lzer$ with the  topology of  convergence in
probability.
\end{theorem}
In the special case of complete markets, we have the following
converse of Theorem \ref{thm:main}.
\begin{proposition} \label{pro:converse}
Let $\{\ld^n\}_{n\in\N_0}$ be a sequence in $\LM$ such that each
  $\ld^n$ defines a complete market, i.e.,
  $\MM^{\ld^n}=\set{\QQ^{\ld^n}}$.
   Suppose
  that $\uln(x)\to u^{\ld^0}(x)$ and $\hXxln_T\to\hat{X}^{x,\lambda^0}_T$ in probability,
  for all $x>0$. Then the sequence $\{Z^{\ld^n}_T\}_{n\in\N}$ is
  $V$-relatively compact and $\ld^n\to\ld^0$ in all appropriate
  topologies.
\end{proposition}
\subsection{On the conditions in the main Theorem \ref{thm:main}}
The purpose of this subsection is provide some intuition about
 the requirement of $V$-relative compactness in connection
with Theorem \ref{thm:main} and Proposition \ref{pro:converse}. We
consider an investor whose preferences are of the  power-type, i.e.,
$$
U(x) = \frac1\gamma x^{\gamma}, \quad V(y) = \frac1{\gamma'}
y^{-\gamma'} \text{ where } \gamma' = \frac\gamma{1-\gamma}.
$$
for some $\gamma\in (-\infty,1)\setminus \set{0}$. 
For $\gamma <0$, the $V$-relative compactness property holds
automatically.
For $\gamma \in (0,1)$, however,  this is not always the case.

 Specializing further, let us assume that all the markets 
$\ld\in\Lambda'\subseteq \Lambda$ under
 consideration are complete, and that the value functions in
 \eqref{equ:UMP} are finite. Conver duality theory (also known as the
 martingale method in financial literature) 
relates the optimal terminal wealth
$\hXl_T$ to the state-price-deflator
$\Zl_T$ via
\begin{align}\label{eq:CRRAoptimalwealth1}
(\hXl_T)^{\gamma-1} = y \Zl_T,
\end{align}
where the Lagrange multiplier $y=y(x,\lambda)$ 
corresponding to
the agent's budget constraint is uniquely determined
by the equation $x = \EE[\Zl_T \hXl_T]$,  with $\hXl_T$ as given by
\eqref{eq:CRRAoptimalwealth1}. Indeed, solving for $\hXl_T$
allows us to compute $y$ explicitely:
\begin{align}\label{eq:CRRAoptimalwealth2}
y=\left( \frac{1}{x} \mathbb{E} \left[(\Zl_T)^\frac{\gamma}{\gamma-1}
\right]\right)^{\gamma-1}.
\end{align}
Equation \eqref{eq:CRRAoptimalwealth1} implies that $\hXl_T$ varies
continuously with $\ld$, essentially if and only if the Lagrange
multiplier $y=y(\ld)$ does.  This is, in turn, intimately related 
 to the concept of $V$-relative compactness. More precisely, let  
$\{\Zln_T\}_{n\in\N_0}$,
$n=0,1,2...$, is the sequence of state-price-deflators
characterizing the financial market with $Z^n \to Z^0$ in
probability, we see that the uniform integrability condition of
$\{(Z^n)^\frac{\gamma}{\gamma-1}\}_n$ allows us to to conclude
that $\EE[(Z^n)^\frac{\gamma}{\gamma-1}] \to
\EE[(Z^0)^\frac{\gamma}{\gamma-1}]$, meaning that the Lagrange
multipliers $y^n$ of \eqref{eq:CRRAoptimalwealth2} converges to
$y^0$. To further elaborate on this, recall that since $Z^n \to
Z^0$ in probability, also $(Z^n)^\frac{\gamma}{\gamma-1}\to
(Z^0)^\frac{\gamma}{\gamma-1}$ in probability, however, this can
happen without the corresponding expectations converge,
$\EE[(Z^n)^\frac{\gamma}{\gamma-1}]\nrightarrow\EE[(Z^0)^\frac{\gamma}{\gamma-1}]$,
in which case the Lagrange multipliers do not converge,
$y^n\nrightarrow y^0$. This section is concluded with a detailed
description of this phenomenon. To be specific, we illustrate that
it is possible to construct a sequence $\{\lambda^n\}_{n\in\N}$
converging to $\lambda^0$ in $\mathbb{L}^2(\mu^M)$, hence $Z^n \to
Z^0$ in probability, however, the expectations do not converge,
$\EE[(Z^n)^\frac{\gamma}{\gamma-1}]\nrightarrow\EE[(Z^0)^\frac{\gamma}{\gamma-1}]$.
Consequently, the investor's utility maximization problem is
ill-posed in the sense that both the value function and the
optimal terminal wealths depend discontinuously of the market
price of risk.

\begin{example}
  Let $\Fi= \mathcal{F}_{t\in[0,1]}$ be the augmented
  filtration generated by a single Brownian motion $B$, and let
  $\set{f^n}_{n\in\N}$ be the sequence of positive, $\FF_1$-measurable
  random variables, given by
\begin{align*}
  f^n(\omega) \triangleq
  \begin{cases}
    n & \text{if }B_1(\omega) \ge \alpha_n\\ 1& \text{if } B_1(\omega)
    \in (\beta_n,\alpha_n)\\ n^{-1} & \text{if }B_1(\omega) \le
    \beta_n
  \end{cases}
\end{align*}
 where the increasing sequence
$\{\alpha_n\}_{n\in\mathbb{N}}$ and the decreasing sequence
$\{\beta_n\}_{n\in\mathbb{N}}$ are given implicitly by
$\Phi(\alpha_n)= 1-\tfrac{1}{2}n^{-5}$ and $\Phi(\beta_n)
=\tfrac{1}{2} n^{-3}$, where $\Phi(\cdot)$ denotes the
distribution function of the the standard normal random variable.
It follows by a direct computation that $f^n\to 1$ almost surely
and $\EE[f^n] \to 1$. By the Martingale Representation Theorem and
since $f^n(\omega) \in[n^{-1},n]$, it also follows that there
exist a sequence $\{\ld^n\}_{n\in\N}$ of predictable processes in
$\ltwo(\PP\times \Leb)$ ($\Leb$ denotes the Lebesgue measure on
$[0,1]$), such that
\begin{equation}
   \nonumber
   \begin{split}
     d\Zln_t= - \Zln_t \ld^n_t\, dB_t,\ t\in [0,1],\text{ and }
     \Zln_1=c_n f^n,\text{ a.s.},
   \end{split}
\end{equation}
where $c_n=1/\EE[f^n]$, so that $\Zln_0=1$, for all $n\in\N$. The
financial market with the risky asset $\Sln$, where
\begin{equation}
   \nonumber
   \begin{split}
d\Sln_t=\Sln_t(\ld^n_t\, dt+dB_t),\ t\in (0,1],\ \Sln_0=1,
   \end{split}
\end{equation}
admits an equivalent martingale measure $\QQ^n$ with
$\RN{\QQ^n}{\PP}=\Zln_1$. By the It\^o-isometry we have
\begin{equation}
   \nonumber
   \begin{split}
     \norm{\ld^n}^2_{\mathbb{L}^2(\PP\times\Leb)} &= \EE\left[\int_0^1 \left(\ld^n_u\right)^2\, du\right]\leq \EE\left[\int_0^1 \left(\ld^n_u\right)^2 n^2 \left(\Zln_u\right)^2 du\right]\\
     &= n^2 \EE\left[\left( \int_0^1 \ld^n_u \Zln_u\, dB_u
       \right)^2\right] = n^2
     \EE\left[\left(\Zln_1-1\right)^2\right]\\
&= n^2\EE\left[\left(c_nf^n-1\right)^2\right]=n^2\Big\{ (n c_n-1)^2(1-\Phi(\alpha_n))+\\
   &+(c_n-1)^2(\Phi(\alpha_n)-\Phi(\beta_n))+(n^{-1}c_n-1 )^2
     \Phi(\beta_n)\Big\}\to 0,
   \end{split}
\end{equation}
by the construction of $\alpha_n$ and $\beta_n$, and thanks to the
fact that $c_n\to 1$. Thus, $\ld^n\to 0$ in $\ltwo(\PP\times\Leb)$
and $\Zln_1\to 1$ in $\ltwo(\PP)$ and in probability, showing that
$\ld^n\to\ld^0\equiv 0$ {\em appropriately} (see Definition
\ref{def:appropriate} and Example \ref{exa:appropriate}).

The optimal terminal wealth $\hXln_1$, in the market with the risky
asset $\Sln$, and for an investor with unit initial wealth and the
power utility $U_{3/4}(x)=\tfrac{4}{3} x^{3/4}$, is given by the
first order condition $U'(\hXln_1) = y_n\Zln_1$ or equivalently
\begin{equation}
   \nonumber
   \begin{split}
\hXln_1= y_n^{-4} (\Zln_1)^{-4},
   \end{split}
\end{equation}
where $y_n>0$ is the Lagrange multiplier determined by the budget-constraint
\begin{equation}
   \nonumber
   \begin{split}
     1=\EE^{\QQ^n}[\hXln_1]=y_n^{-4}\EE\left[\left(c_nf^n\right)^{-3}\right].
   \end{split}
\end{equation}
An explicit computation yields
$$
y_n^4 = c_n^{-3}\Big\{ n^{-3}(1-\Phi(\alpha_n))
+(\Phi(\alpha_n)-\Phi(\beta_n))+n^3 \Phi(\beta_n)\Big\}\to \frac32.
$$
Since
$\Zln_1 \to 1$ in probability,
the sequence $\hXln_1$ converges in probability towards the constant
random variable with value $\frac23$. On the other hand, the optimal
strategy in the limiting market (where the risky security evolves as
$dS_t=S_t\, dB_t$), is not to invest in the risky asset at all,
making $\hX_1=1$ the optimal terminal wealth. It is clear now that
no convergence of the optimal terminal wealths can take place, even
though the convergence $\ld^n\to\ld^0=0$ is appropriate, and even in
$\ltwo(\PP\times\Leb)$.  One could obtain a number of similar
counterexamples (oscillatory behavior, convergence of the Lagrange
multipliers to $+\infty$ or to $0$) by a different choice of
parameters.
\end{example}

\section{Proofs}
The strategy behind the proof of our main Theorem \ref{thm:main} is
to place the utility-maximization problem \eqref{equ:UMP} in an
appropriate functional-analytic framework and to exploit the dual
representation of the value function $\ul$ and the optimal terminal
wealth $\hXxl_T$. The steps of this program are the content of this
section and some of the techniques we apply are inspired by the
proof of Berge's Maximum Theorem.

\subsection{The dual approach to utility maximization}
The results  of \cite{KraSch99}  guarantee the existence and
uniqueness of the optimal terminal wealth in each market $\Sl$,
$\ld\in\LM$, under mild regularity conditions. Moreover, building
on the work of \cite{KarLehShrXu91} and others, the authors of
\cite{KraSch99} have established a strong duality relationship
between the primal utility-maximization problem \eqref{equ:UMP}
and a suitable dual problem posed over the set of martingale
measures $\MMl$, or its enlargement $\YYl$. It is this last
formulation that is most suited for our purposes. More precisely,
with the dual value function $\vl$ being defined by
\begin{equation}
   \label{equ:DVF}
   \begin{split}
     v^{\ld}(y)=\inf_{\QQ\in\MMl} \EE[V(y\dqp)],
   \end{split}
\end{equation}
the main result of \cite{KraSch99}  is the content in Theorem
\ref{thm:Kramkov-Schachermayer} below. We state it for the
reader's convenience, since its content will be used extensively
in the sequel.
\begin{theorem}[Kramkov, Schachermayer, \dots]
\label{thm:Kramkov-Schachermayer}
Let $\ld\in\LM$ be arbitrary, but fixed, and let $u^{\ld}(\cdot)$ and
$v^{\ld}(\cdot)$ be the value functions of the primal and the dual
problem defined above in \eqref{equ:UMP} and \eqref{equ:DVF}. Then, if
$\ul(\cdot)$ does not
identically equal $+\infty$, the following statements hold:
\begin{itemize}
\item[(a)] Both $\ul:(0,\infty)\to\R$ and $\vl:(0,\infty)\to\R$ are finite
  valued, and continuously
  differentiable. Furthermore, $\ul$ is strictly concave and increasing, $\vl$ is
  strictly convex and decreasing, and  the
following conjugacy relation holds between
  them
\begin{equation}
   \nonumber
   \begin{split}
\vl(y)=\sup_{x>0} \big( \ul(x)-xy \big),\ \forall\, y>0.
   \end{split}
\end{equation}
\item[(b)] Alternatively, the dual value function is given as
\begin{equation}
   \label{equ:alternative-vl}
   \begin{split}
\vl(y)=\inf_{Y\in \YYl} \EE[V(yY_T)],
   \end{split}
\end{equation}
where the enlarged domain $\YYl$ is the set of all non-negative \cd
supermartingales $Y$ with $Y_0=1$ such that $X Y$ is a supermartingale
for each $X\in\XXl(1)$. The infimum in \eqref{equ:alternative-vl} is
uniquely attained in $\YYl$ (with the minimizer denoted by $\hYyl$).
\item[(c)] For $x>0$ and $y=(\ul)'(x)$,
 the random variable $\hXxl_T=-V'(y\hYyl_T)$,
belongs to $\XXl(x)$ and is the a.s.-unique
  optimal terminal wealth for an agent with initial wealth $x$ and utility function $U$.
\end{itemize}
\end{theorem}

\subsection{Structure of the dual domain}
Some of the central arguments in the proof of our main result
\ref{thm:main} depend on a precise characterization of the set
$\YYl$ introduced in (b) above. Thanks to the continuity of the
paths of our price process $\Sl$, this can be achieved in a quite
explicit manner, as described in the following proposition.
\begin{proposition}
\label{pro:structure-of-supermartingales} For $\ld\in\LM$, let $Y$
be in $\YYl$, i.e., $Y$ is a non-negative \cd supermartingale such
that $Y_0=1$, and $YX$ is a supermartingale for each $X\in \XXl(1)$.
When $Y_T>0$ a.s., we have the following multiplicative
decomposition
$$
Y=\Zl \EN(L) D,
$$
where $\Zl=\EN(-\ld\cdot M)$, L is a \cd local martingale,
strongly orthogonal to $M$, meaning $\langle M,L\rangle \equiv 0$,
and $D$ is a predictable, non-increasing, \cd process with
$D_0=1$, $D_T>0$, a.s.
\end{proposition}

\begin{proof}
For the sake of notational clarity, we omit the superscript $\ld$
from all expressions in the present proof. Since $Y$ is strictly
positive, $Y$ has a multiplicative Doob-Meyer decomposition:
$$
Y_t = \mathcal{E}(-\alpha\cdot M + L)_t D_t
$$
for some $\alpha\in L(M)$, a local martingale $L$ satisfying
$\langle L,M\rangle\equiv0$,
 and a predictable,
c\'{a}dl\'{a}g, non-increasing process $D$ (see Theorem 8.21,
p.~138 in \cite{JacShi03}). Thanks to the strong orthogonality of
$L$ and $M$, the relationship $\EN(-\alpha\cdot
M+L)=\EN(-\alpha\cdot M)\,\EN(L)$ holds. Therefore, it remains to
show that $\alpha = \lambda$ almost everywhere with respect to the
measure $\mu^M$ defined in \eqref{equ:mu-A}.

By the strict positivity of the process $D$, we can write
$dD_t=D_{t-}dF_t$ for a non-increasing predictable process $F$.
 Using Theorem 2.1 in \cite{DelSch95}, $F$ can be split into an
 integral with respect to $d\ab{M}$ (the absolute continuous part) and a
singular part $F'$. More precisely,  there exists a $\mu^M$-null set $A$
with $F' = \int_0^{\cdot} 1_A(u)dF'_u$, and
a non-negative predictable process $\beta$ such that
$$
F_t = -\int_0^t\beta_u\, d\ab{M}_u + F'_t,\ t\in [0,T]
$$
With this notation we have $dD_t = -D_{t-}\beta_td\langle
M\rangle_t+D_{t-}dF'_t$ and by It\^{o}'s Lemma and the
predictability of $F$ we get
$$
dY_t = Y_{t-}(-\alpha_t\, dM_t + dL_t-\beta_t\,d\ab{M}_t +
dF'_t),\ t\in (0,T],\ Y_0=1.
$$
Therefore for any admissible portfolio wealth process
$X\in\mathcal{X}^\lambda(1)$ generated by a portfolio $H$ we have
\begin{align*}
d(Y_tX_t)
= Y_{t-}H_t(\lambda_td\langle M\rangle_t+dM_t) &+
X_tY_{t-}(-\alpha_tdM_t + dL_t-\beta_td\langle M\rangle_t +dF'_t)-\\
&-Y_{t-}\alpha_tH_td\langle M\rangle_t
\end{align*}
and given the supermartingale property the drift in the above has
to be non-positive, meaning that for any $H$ we have the
inequality
$$\Big( H_t( \lambda_t-\alpha_t)-\beta_tX_t\Big)d\langle
M\rangle_t+X_tdF'_t\le 0,$$ in the sense that the measure the
left-hand-side generates on the predictable sets is non-positive.
Moreover, by the singularity between $\mu^M$ and $dF'$, the
following must hold $\mu^M$-a.e.
$$H_t( \lambda_t-\alpha_t)\le \beta_tX_t$$
for all admissible $H$.  Suppose now, contrary to the claim we are
trying to prove, that
$\mu^M(\ld\not=\alpha)>0$. Without loss of generality we
assume that this implies that
exists a predictable set $A_1\subseteq [0,T]\times \Omega$ with the
property that
\begin{enumerate}
\item $\ld-\alpha\geq \eps$ on $A_1$  for some $\eps>0$, and
\item $\mu^M(A_1)>0$.
\end{enumerate}

Since  $\beta$ and $\ld$ are finite-valued predictable
process and $\beta$ is non-negative,
 we can find a constant $\Sigma>0$
and a predictable set $A_2$ such that $\beta,\abs{\ld}\in [0, \Sigma]$
on $A_2$ and $\mu^M (A)>0$, where $A=A_1\cap A_2$.

For $n\in\N$, let $\tilde{H}^n$ be the predictable process given by
$\tilde{H}=n \ind{A}$, and let $\tau_n$ be the first exit
 time
of the process $1+\tilde{H}^n\cdot S$ from the semi-axis
$(0,\infty)$. Define the adjusted predictable process $H^n$ by
 $H^n=\tilde{H}^n\ind{[0,\tau_n]}$, so  that
 $\mu^M(\set{H^n>0})>0$.
For each $n$, $H^n$ is predictable and $X^n \triangleq 1+(H^n\cdot
S)$ is in $\mathcal{X}(1)$ and so by the above we have
\begin{equation}
   \label{equ:pre-central-thing}
   \begin{split}
\eps/\Sigma &\leq  X^n \beta_t/\Sigma \leq X^n=(1+n\ind{A}\cdot S)^{\tau_n},
\ \mu^M-\text{a.e. on $A$.}
   \end{split}
\end{equation}
Observe that one of the conclusions of \eqref{equ:pre-central-thing}
is that the stopping time $\tau_n$ will not be realized on $A$
(because
the process $1+n\ind{A}\cdot S$ is continuous). Also, as the process
$\ind{A}\cdot S$ is constant off $A$, we have
 the following strengthening of \eqref{equ:pre-central-thing}
\begin{equation}
   \label{equ:central-thing}
   \begin{split}
\eps/\Sigma \leq 1+n\ind{A}\cdot S, \mu^M-\text{a.e.}
   \end{split}
\end{equation}
Define the non-decreasing continuous process
$C$ by $C_t=\int_0^t \ind{A}(u)\, d\ab{M}_u$, and note that
$\mu^M(A)>0$ implies that $\PP[C_T>0]>0$.
 Therefore,
 the right inverse $G$ of $C$, given by
$G_s=\inf\sets{t\geq 0}{C_t> s}$, where $\inf\emptyset=+\infty$,
is a right-continuous, non-decreasing $[0,\infty]$-valued stochastic
process, such that $G_s<\infty$ on the (non-trivial) stochastic
interval $[0,C_T)$. Define the  process $V$ by
\begin{equation}
   \nonumber
  \begin{split}
V_s=\begin{cases}
S_{G_s}-S_0,& \text{when $G_s<\infty$, and} \\
S_{T}-S_0+\tilde{B}_{s-C_T}& \text{otherwise,}
\end{cases}
   \end{split}
\end{equation}
where $\tilde{B}$ is a Brownian motion, defined on an extension of the
probability space $(\Omega,\FF,\PP)$ and independent of $\FF_T$.
An application of L\' evy's criterion shows that $V$ is
a Brownian motion with drift $\ld_{G_s}\inds{G_s<\infty}$.
Letting $n\to\infty$ in \eqref{equ:central-thing} yields that
$V_s\geq 0$ for $s\in [0,C_T)$. On the other hand, as $\ld\leq \Sigma$
on $A$, $V$ is bounded from above by a Brownian motion with a constant
drift $C$. This is, however, a contradiction, as almost every
trajectory of a Brownian motion with a constant drift enters the negative
semi-axis $(-\infty,0)$, in every neighborhood of $0$.
\end{proof}

\begin{corollary}
  \label{cor:inf-attained-at-a-local-martingale}
For  $\ld\in\Lambda_M$, $\vl$ is the value function of the dual
optimization problem defined by \eqref{equ:DVF}. For each $y>0$,
such that $\vl(y)<\infty$, there exists a local martingale $\Lyl$,
strongly orthogonal to $M$, such that
\begin{equation}
   \nonumber
   \begin{split}
\vl(y)=\EE[V(y\Zl_T \EN(\Lyl))].
   \end{split}
\end{equation}
\end{corollary}
\begin{proof}
Theorem \ref{thm:Kramkov-Schachermayer}, (c)  implies that the
infimum in the definition of  the dual value function $\vl$ is
attained at a terminal value $Y_T$ of a supermartingale $Y$ with the
property that $YX$ is a supermartingale for each $X\in\XXl(1)$.
Proposition \ref{pro:structure-of-supermartingales} states that each
such supermartingale can be written
$$
Y = Z^\lambda \EN(L)D,\text{ with $\langle L,M\rangle\equiv 0$.}
$$
Thanks to the strictly decrease of the function $V$,
 we must have $D\equiv 1$. Indeed, $Z^{\ld} \EN(L)$ dominates $Z^{\ld} \EN(L) D$ pointwise,
and belongs to $\YY^{\ld}$.
\end{proof}
Let $\BB$ denote the set of all local martingales $L$, strongly
orthogonal to $M$,  such that the terminal value $\EN(L)_T$ of the
stochastic exponential $\EN(L)$ is bounded from below by a positive
constant.
\begin{corollary}
\label{cor:infimum-over-bounded} Let $\lambda\in\Lambda_M$ and
suppose that $\EE[V^+(\Zl_T)]<\infty$. Then for each $y>0$ we have
the representation
\begin{equation}
   \label{equ:infimum-over-bounded}
   \begin{split}
\vl(y)=\inf_{L\in\BB} \EE[ V(y\Zl_T \EN(L)_T)].
   \end{split}
\end{equation}
\end{corollary}
\begin{proof}
The fact that the infimum on the right-hand side of
\eqref{equ:infimum-over-bounded} is bounded from below by the
value function $\vl(y)$ follows directly from Theorem
\ref{thm:Kramkov-Schachermayer} and Proposition
\ref{pro:structure-of-supermartingales}. For the other inequality,
let $\Lyl$ be the local martingale from the statement of Corollary
\ref{cor:inf-attained-at-a-local-martingale}. If $\EN(\Lyl)_T$
happened to be bounded from below by a strictly positive constant,
there would be nothing else left to prove. However, $\EN(\Lyl)_T$
is, in general, not bounded away from zero, so we employ a
limiting argument via a suitably defined sequence $L^n\in \BB$.
For $n\in\N$, let $Y^n$ be the supermartingale in $\YYl$ given by
\[Y^n=\Zl\big(\tfrac{n-1}{n} \EN(\Lyl)+\tfrac{1}{n}\big).\]
The process $Y^n$ is a positive local martingale with the property
that $Y^n X$ is a supermartingale for each $X\in\XXl(1)$ and
therefore the proposition allows us to write $Y^n =
Z^\lambda\EN(L^n)D^n$ and since $D^n\le 1$ we have
$\EN(L^n)\ge\frac1n$ and so $\EN(L^n)\in\mathcal{B}$. Furthermore,
since $V$ is decreasing and convex we have
\begin{align*}
\EE[V(yZ_T^\lambda\EN(L^n)_T)] &\le \EE[V(yY^n_T)]\\&\le
\frac{n-1}{n}\EE\left[V\Big(yZ^\lambda_T\EN(\Lyl)_T\Big)\right]
+\frac{1}{n}\EE\left[V\Big(yZ_T^\lambda\Big)\right]\\
&\le\frac{n-1}{n}v^\lambda(y)
+\frac{1}{n}\EE\left[V^+\Big(yZ^\lambda_T\Big)\right]\\
&\le\frac{n-1}{n}v^\lambda(y)
+\frac{1}{n}\left(C\,\EE\left[V^+\Big(Z^\lambda_T\Big)\right]+D\right)
\end{align*}
for two constants $C$ and $D$ granted by the asymptotic elasticity
of $U$ (see Proposition 6.3(iii) of \cite{KraSch99}). Taking  the
$\liminf$ with respect to $n$ on both sides yields the desired inequality.
\endproof

\end{proof}

\subsection{Joint continuity of the value functions}
The following lemmas establish a joint continuity property for the
primal and dual value functions and their derivatives. Before we
proceed, let us agree that in the sequel $\Lambda'\subseteq
\Lambda_M$ is $V$-relatively compact, and that $\tau$ is an
appropriate topology. By the inequality
\begin{equation}
   \nonumber
   \begin{split}
U(X_T)\leq V(\Zl_T)+X_T\Zl_T,
   \end{split}
\end{equation}
and the supermartingale property of the process $X\Zl$ when
$X\in\XX^\lambda(x)$,  it follows that $\ul(x)\leq
\EE[V(\Zl_T)]+x<\infty$, for all $x>0$ and $\ld\in\Lambda'$.
Therefore, the assumptions of the Theorem
\ref{thm:Kramkov-Schachermayer} are satisfied, and its conclusions
hold.
\begin{lemma}
\label{lem:continuous-in-probability} Let $Y$ be a random variable,
bounded from below by a strictly positive constant, such that
$\sup_{\ld\in\Lambda'}\EE[\Zl_T Y]<\infty$. Then the mapping
$(y,\ld)\mapsto V(y \Zl_T Y)$ is continuous from $(0,\infty)\times
\Lambda'$ (with the product topology) into $\lone$. In particular, the
mapping $(y,\ld)\mapsto \EE[V(y\Zl_T Y)]$ is continuous.
\end{lemma}
\begin{proof}
Given that $V$ is a continuous function,  the mapping
 $(y,\ld)\mapsto V(y\Zl_T Y)$
is continuous in probability because  $(y,\ld)\mapsto y \Zl_T$ is. It
will, therefore,  be
enough to establish the uniform integrability of the family
\begin{equation}
   \nonumber
   \begin{split}
\sets{ V(yY \Zl_T)}{y\in B, \ld\in\Lambda'},
   \end{split}
\end{equation}
when $B$ is a compact segment of the form $[\eps,1/\eps]$, $\eps>0$.
The boundedness in $\lone$ of the
family $\sets{y Y\Zl_T}{y\in B,\ \ld\in\Lambda'}$ and
the fact that $\lim_{y\to\infty} \tfrac{V^-(y)}{y}=0$, coupled
with the De la Vall{\' e}e Poussin  criterion, imply that
the family $\sets{V^-(y Y\Zl_T)}{y\in B,\ld\in\Lambda'}$ is
uniformly integrable. As for
the positive parts, it will be enough to note that
$V^+(y Y\Zl_T)\leq V^+(y_0 \Zl_T)$, where $y_0=\eps \essinf Y>0$,
and invoke the argument concluding
the proof of Corollary \ref{cor:infimum-over-bounded} to reach the
conclusion that the positive parts
\[\sets{V^+(y Y\Zl_T)}{y\in B, \ld\in\Lambda'}\]
form a uniformly integrable family as well.
\end{proof}
\begin{lemma}
\label{lem:v-upper-semicontinuous}
The function
\begin{equation}
   \nonumber
   \begin{split}
(y,\ld)\mapsto \vl(y),
   \end{split}
\end{equation}
mapping $(0,\infty)\times \Lambda'$ into $\R$ is upper semi-continuous
(with respect to the product topology).
\end{lemma}
\begin{proof}
By Corollary \ref{cor:infimum-over-bounded}
 the dual value function $\vl$ has the
following representation
\begin{equation}
   \nonumber
   \begin{split}
\vl(y)=\inf_{Y} \EE[V(y Y \Zl_T)],
   \end{split}
\end{equation}
where the infimum is taken over $Y$ of the form $Y=\EN(L)$, where
$L\in\BB$, i.e.\ $\EN(L)$ is bounded away from zero. For a such a
random variable $Y$, by Lemma \ref{lem:continuous-in-probability},
the mapping $(y,\ld)\mapsto \EE[V(yY\Zl_T)]$ is continuous.
Therefore, $(y,\ld)\mapsto\vl(y)$  is $\tau$-upper semi-continuous
as an infimum of continuous mappings.
\end{proof}

\begin{lemma}
\label{lem:dual-value-functions-converge} The mapping
$(y,\ld)\mapsto \vl(y)$ is continuous on
$(0,\infty)\times\Lambda'$ (with respect to the product topology).
\end{lemma}
\begin{proof}
Thanks to the result of Lemma \ref{lem:v-upper-semicontinuous}, it
is enough to show that $(y,\ld)\mapsto \vl(y)$ is lower
semi-continuous. Let $\seeq{y_n,\ld^n}$ in $(0,\infty)\times
\Lambda'$ converge to $(y,\ld)\in (0,\infty)\times\Lambda'$ and we
need to prove that $\vl(y)\leq \liminf \vln(y_n)$. By passing to a
subsequence that realizes the liminf, we can assume that the
sequence $\seeq{\vln(y_n)}$ converges and, furthermore, by passing
to yet another subsequence we may assume that $Z^{\lambda^n}_T\to
Z^{\lambda}_T$ almost surely.

Corollary \ref{cor:inf-attained-at-a-local-martingale} states that
$\vln(y_n)=\EE[V(y_n \hYyln_T)]$, where the  optimizer $\hYyln_T$
can be written as $\hYyln_T= \Zln_T \EN(L^n)_T$, for some local
martingale $L^n=L^{y_n,\lambda^n}$ that is strongly orthogonal to
$M$. Komlos' lemma grants the existence of an almost surely
convergent sequence in $conv\left(y_n Z_T^{\lambda^n}
\EN(L^n)_T,y_{n+1} Z_T^{\lambda^{n+1}} \EN(L^{n+1})_T,...\right)$.
So there exist a double array $\{\alpha^n_k\}$ with $n\in\N$, $k
\in \{n,...,K(n)\}$ for some $K(n)\in \N$, of positive weights and
a random variable $h\in\lzer_+$ such that
\begin{equation}
    \nonumber
    \begin{split}
\sum_{k=n}^{K(n)} \alpha^n_k=1,\text{ for all $n$, and }
h_n=\sum_{k=n}^{K(n)} \alpha^n_k y_kZ^{\lambda^k}_T \EN(L^k)_T\to
h,
    \end{split}
\end{equation}
where the convergence is $\PP$-almost surely. Since also
$y_nZ^{\lambda^n}_T \to yZ^{\lambda}_T$ almost surely, we have
(see Lemma \ref{lem:mini} below)
\begin{equation}
    \nonumber
    \begin{split}
f_n=\sum_{k=n}^{K(n)} \alpha^n_k \EN(L^k)_T \to
\frac{h}{yZ},\text{ almost surely.}
    \end{split}
\end{equation}
The random variables $f_n$ are all in $\YY^0=\YY^{\ld\equiv
  0}$, which is closed with respect to convergence in probability,
thanks to Lemma 4.1., p.~926 in \cite{KraSch99}. Therefore, the
limit of $f_n$ will also be in $\YY^0$, and, consequently,
$\frac{h}{y}\in\YYl$. By Fatou's Lemma (keeping in mind the
uniform integrability of the family of negative parts
$\sets{V^-(Y)}{Y\in\lzer_+,\,\EE[Y]\leq c}$ for any $c>0$) we have
\begin{equation}
    \nonumber
    \begin{split}
\vl(y)&\leq \EE[V(h)]=\EE[V(\liminf_n h_n)]\\
&\leq \liminf_n  \EE\left[V\left(\sum_{k=n}^{K(n)} \alpha^n_k y_kZ^{\lambda^k}_T \EN(L^k)_T\right)\right]\\
&\leq \liminf_n \sum_{k=n}^{K(n)} \alpha^n_k \EE\left[V\left(y_k \hYylk_T\right)\right]\\
&=\liminf_n \sum_{k=n}^{K(n)} \alpha^n_k v^{\lambda^k}(y_k) =
\lim_n \vln(y_n).
    \end{split}
\end{equation}

\end{proof}

\begin{lemma}\label{lem:mini} Let $\{a_n\}_{n\in\N}$ be a sequence
of positive numbers converging to $a>0$ and assume
$\sum_{k=n}^\infty a_kb_k^n$ converges to $c>0$ where $\{b_k^n\}$
is some double array of positive numbers. Then $\sum_{k=n}^\infty
b^n_k$ converges to $c/a$.
\end{lemma}

\begin{proof}
Let $\epsilon >0$ be arbitrary. We can find $N(\epsilon)\in\N$
such that $(1+\epsilon) \ge a_k/a \ge (1-\epsilon)$ for all $k\ge
N(\epsilon)$. Therefore, for $n\ge N(\epsilon)$ we have
$$
\frac1{a(1+\epsilon)}\sum_{k=n}^\infty a_kb^n_k \le
\sum_{k=n}^\infty b^n_k \le \frac1{a(1-\epsilon)}\sum_{k=n}^\infty
a_kb_k^n
$$
and hence, passing $n$ to infinity yields the desired conclusion.
\end{proof}

\begin{proposition}
\label{pro:value-functions-converge}
The following mappings are continuous on  $(0,\infty)\times\Lambda'$
\begin{align*}
     (y,\ld) & \mapsto \vl(y), &
     (y,\ld) & \mapsto (\vl)'(y), &
     (x,\ld) & \mapsto \ul(x), &
     (x,\ld) & \mapsto (\ul)'(x).
   \end{align*}
 \end{proposition}
\begin{proof}
Let $\{\ld^n\}_{n\in\N}$ be a sequence in $\Lambda'$ converging
appropriately to $\ld\in\Lambda'$. Thanks to the result of Lemma
\ref{lem:dual-value-functions-converge} and the convexity of the
dual value functions, Theorem 25.7 in \cite{Roc70} states that the
derivatives $(\vln)'(\cdot)$ converge towards $(\vl)'(\cdot )$, as
well, uniformly on compact intervals in $(0,\infty)$. The uniform
convergence on compact intervals also holds for the original
sequence of functions $\vln(\cdot)$.

To proceed, pick $x>0$ and $\eps>0$, and define
$y(\epsilon)\triangleq (\ul)'(x)+\epsilon$.
The strict increase of $(\vl)'(\cdot)$ implies that
\begin{equation*}
\lim_n (\vln)'(y(\epsilon))= (\vl)'(y(\epsilon))=
(\vl)'((\ul)'(x)+\epsilon)>(\vl)'((\ul)'(x))=-x,
\end{equation*}
where the last inequality follows directly from continuous
differentiability and conjugacy of $\ul$ and
$\vl$. Consequently, for large $n$,
 we have $-(\vln)'(y(\epsilon)) < x$.
Since $(\uln)'(\cdot)$ is strictly decreasing for each $n\in\N$, we get
\begin{equation*}
(\ul)'(x) + \epsilon = y(\epsilon)= (\uln)'\Big(
- (\vln)'(y(\epsilon))\Big) > (\uln)'(x),
\end{equation*}
for large $n$, implying that $\limsup_n (\uln)'(x) \leq
(\ul)'(x)$. The other inequality, namely $\liminf_n (\uln)'(x)
\geq (\ul)'(x)$, can be proved similarly. By the results obtained
so far we have
\begin{equation*}
\uln(x) = \vln\Big((\uln)'(x)\Big) + x\,(\uln)'(x)\to
\vl\Big((\ul)'(x)\Big) + x\,(\ul)'(x)= \ul(x).
\end{equation*}
Finally, the joint continuity of value functions and their
derivatives on $\Lambda'\times (0,\infty)$ is a consequence of the
already mentioned uniform convergence from Theorem 25.7 in
\cite{Roc70}.
\end{proof}
\subsection{Continuity of the optimal terminal wealths}
\begin{lemma}
\label{lem:continuity-of-terminal-wealths}
Let $\Lambda'\subseteq\Lambda_M$ be $V$-relatively compact
and let $\tau$ be an appropriate topology on $\Lambda'$.
The function
\begin{equation}
   \nonumber
   \begin{split}
(x,\ld)\mapsto \hXxl_T,
   \end{split}
\end{equation}
where $\hXxl_T$ is the unique optimal terminal wealth in the market $\Sl$,
is continuous from $(0,\infty)\times \Lambda'$ to $\lzer$ (equipped with the
topology of convergence in probability).
\end{lemma}
\begin{proof}
By Theorem \ref{thm:Kramkov-Schachermayer}, the optimal terminal
wealth admits the representation with $y=(\ul)'(x)$,
$$
U'(\hXxl_T)=y \hYyl_T
$$
where $\hYyl_T$ attains the minimum in the dual problem
\eqref{equ:alternative-vl}. Thanks to the continuity of the
mappings $(x,\ld)\mapsto (\ul)'(x)$ and $x\to U(x)$ it suffices to
show that $(y,\ld)\mapsto y \hYyl_T$ is continuous in probability.
Since the topology $\tau$ is assumed to be metrizable, it is
enough to show that convergence of any sequence $(y_n,\ld^n)\in
(0,\infty)\times \Lambda'$ to $(y,\ld)\in (0,\infty)\times
\Lambda'$ implies the convergence of $y_n \hYyln_T\to y \hYyl_T$
in probability. By Proposition
\ref{pro:structure-of-supermartingales}, each $\hYyln_T$ can be
expressed as the product
\begin{equation}
   \nonumber
   \begin{split}
\hYyln_T=Z^n H^n,\text{ where $Z^n=\Zln_T$ and $H^n=\EN(\Lyln)_T$,}
   \end{split}
\end{equation}
for some local martingale $\Lyln$, strongly orthogonal to $M$. In
the same way, we will write $\hYyl_T=Z H$, with analogous
definitions for $Z$ and $H$.

We start our analysis by noting that for any $\delta>0$, with
$H^{\delta}=(1-\delta)H+\delta$, we have by Markov's inequality
\begin{equation*}
   \begin{split}
\PP[ y_n Z^n |H^n-H|> 2 \eps ] & \leq \PP[ y_n Z^n |H^n-H^{\delta}|>\eps]+ \PP[ y_n Z^n \delta |1-H|>\eps]\\
&\leq \PP[ y_n Z^n |H^n-H^{\delta}|>\eps]+\tfrac{\delta}{\eps} \EE[ y_n Z^n |1-H|]\\
&\leq \PP[ y_n Z^n |H^n-H^{\delta}|>\eps]+\tfrac{2y_n \delta}{\eps}
   \end{split}
\end{equation*}
since $\EE[Z^n|1-H|]\leq \EE[Z^n]+\EE[Z^n H]\leq 2$. We then pick a
constant $N>0$ and by the strict concavity of $V$ we can find a
positive constant $\beta=\beta(\eps,N)$ with the property that
\begin{equation*}
V\left(\frac{a+b}{2}\right) < \frac{V(a)+V(b)}2 -\beta.
\end{equation*}
if $a$ and $b$ are positive numbers with $|a-b| > \epsilon$ and $(a
+ b)\le N$. This property, combined with the convexity of $V$, leads
to the following estimate
\begin{multline}
\label{equ:central-inequality} \EE\left[ V\Big(y_nZ^n\frac{H^n +
H^{\delta}}{2} \Big)\right]
-\frac{1}{2} \Big( \EE[V(y_n Z^n H^n)]+\EE[V(y_n Z^n H^{\delta})] \Big) \\
\leq -\beta \PP[ y_n Z^n |H^n-H^{\delta}|>\eps,\ y_n Z^n (H^n+H^{\delta})<N].
\end{multline}
The convex combination $\tfrac{1}{2} Z^n (H^n+H^{\delta})$ belongs
to the dual domain $\YY^{\ld^n}$, and so the following inequality
for the first term on the left-hand side of
\eqref{equ:central-inequality} holds
\begin{equation}
   \nonumber
   \begin{split}
\EE\left[ V\Big(y_nZ^nH^n\Big)\right]= v^{\lambda^n}(y_n)\leq
\EE\left[ V\Big(y_nZ^n\frac{H^n + H^{\delta}}{2} \Big)\right].
   \end{split}
\end{equation}
Combining this estimate with \eqref{equ:central-inequality} gives
\begin{multline*}
\beta \PP[ y_n Z^n |H^n-H^{\delta}|>\eps,\ y_n Z^n (H^n+H^{\delta})<N]  \leq \\
\leq \tfrac{1}{2} \Big ( \EE[V(y_n Z^n H^{\delta})] - v^{\lambda^n}(y_n)\Big)
\end{multline*}
which combined with Markov's inequality grants the inequality
\begin{equation}
   \nonumber
   \begin{split}
\beta \PP[ y_n Z^n &|H^n-H^{\delta}|>\eps] \\& \leq \tfrac{1}{N} \beta \EE[ y_n Z^n |H^n+H^{\delta}|]
 +\tfrac{1}{2} \Big( \EE[V(y_n Z^n H^{\delta})] - v^{\lambda^n}(y_n)\Big) \\
& \leq 2\beta  \tfrac{y_n}{N} + \tfrac{1}{2} \Big ( \EE[V(y_n Z^n
H^{\delta})] - v^{\lambda^n}(y_n)\Big).
   \end{split}
\end{equation}
We therefore have the overall estimate
\begin{equation*}
   \begin{split}
 \PP[ y_n Z^n |H^n-H|> 2 \eps ] & \leq \tfrac{2 y_n \delta}{\eps}+2\tfrac{y_n}{N}+
\tfrac{1}{2\beta} \Big ( \EE[V(y_n Z^n H^{\delta})] - v^{\lambda^n}(y_n)\Big).
   \end{split}
\end{equation*}
By Lemma \ref{lem:continuous-in-probability}, the third term in on the right-hand
side converges to $\EE[ V(y Z H^{\delta})]$ whereas the fourth term
converges to $\vl(y)$, thanks to Proposition \ref{pro:value-functions-converge}, and so the limit,
as $n\to\infty$, of those two terms can be bounded from above by
\begin{equation}
   \nonumber
   \begin{split}
\EE[ V(y Z H^{\delta})]-\EE[V(y Z H)] \leq \delta K, \text{ where } K=\EE[V(yZ)]-\EE[V(yZH)],
   \end{split}
\end{equation}
by a straightforward use of $V$'s convexity. To recapitulate, we
have
\begin{equation}
   \label{equ:recapitulate}
   \begin{split}
\limsup_{n\to\infty} \PP[ y_n Z^n |H^n-H|> 2 \eps ] \leq  \tfrac{2y
\delta}{\eps}+2\tfrac{y}{N}+ K \tfrac{\delta}{2\beta}.
   \end{split}
\end{equation}
Letting first $\delta\to 0$, and then $N\to\infty$ in \eqref{equ:recapitulate} shows that
\begin{equation}
   \nonumber
   \begin{split}
\lim_{n\to\infty} \PP[ y_n Z^n |H^n-H|>\eps] =0,\ \forall\,\eps>0,
   \end{split}
\end{equation}
meaning $y_nZ^nH^n - y_nZ^nH \to 0$ in probability and since also
$y_nZ^n\to yZ$ in probability, the result follows.
\end{proof}
\subsection{Proof of Proposition \ref{pro:converse}}

Convergence of the value functions $\uln$ towards $\ulz$ implies
the convergence of the derivatives $(\uln)'$ (towards $(\ulz)'$)
(see Theorem 25.7 in \cite{Roc70}). Following the ideas from the
proof of Proposition \ref{pro:value-functions-converge}, we can
conclude that $\vln\to\vlz$ and $(\vln)'\to(\vlz)'$. Since
$\hXxln_T=-V'((\uln)'(x) Z^{\ld^n}_T)$, $n\in\N_0$, and
$U'=(-V')^{-1}$ is a continuous function, the sequence
$Z^{\ld^n}_T$ converges towards $Z^{\ld^0}_T$ in probability,
implying that the convergence $\ld^n\to\ld^0$ is appropriate.

By the definition of the dual value functions, we have
 $\vln(y)=\EE[V(y Z^{\ld^n}_T)]$. The
family $\sets{V^-(y Z^{\ld^n}_T)}{\ld\in\Lambda}$ is uniformly
integrable and hence $\EE[V^+(y Z^{\ld^n}_T)]\to \EE[V^+(y
Z^{\ld^0}_T)]$, for every $y>0$. This observation, and the fact
that $V^+(y Z^{\ld^n}_T)\to V^+(y Z^{\ld^0}_T) $ in probability,
can be fed into Scheffe's Lemma to conclude that $\sets{V(y
Z^{\ld^n}_T)}{n\in\N_0}$ is a uniformly integrable sequence.
Setting $y=1$ completes the proof.
\appendix
\section{Appendix}
\begin{proposition}
\label{pro:conv-Fatou}
Suppose that $\ld^n\to\ld^0$ in $\ltwo(\mu^M)$, for some
sequence $\{\ld^n\}_{n\in\N_0}$ in $\ltwo(\mu^M)$. Then
$\Zln_T\to Z^{\ld_0}_T$ in probability.
\end{proposition}
\begin{proof}
The It\^{o}-isometry implies that $\int_0^T \ld^n_u\, dM_u\to
\int_0^T \ld^0_u\, dM_u$ in $\ltwo(\PP)$, and, hence also in
probability. Thanks to the continuity of the exponential function,
it will be enough to to show that
\begin{equation}
\label{to-show}
\int_0^T (\lambda^n_u)^2\, d\ab{M}_u \to \int_0^T (\lambda^0_u)^2\,
d\ab{M}_u \ \text{in probability.}
\end{equation}
Let us recall a well-known characterization of convergence in
probability which states that a sequence $\{X^n\}_{n\in\N}$ of random
variables converges towards a random variable $X^0$ in probability if
and only if for any subsequence $\{X^{n_k}\}_{k\in\N}$ of
$\{X^n\}_{n\in\N}$ there exists a further subsequence
$\{X^{n_{k_l}}\}_{l\in\N}$
which converges to $X^0$ almost surely. With this
in mind, let $\int_0^T (\lambda^{n_k}_u)^2\, d\ab{M}_u$ be an arbitrary
subsequence of $\int_0^T (\lambda^n_u)^2\, d\ab{M}_u$.  Since
$\lambda^{n_k}\to\lambda^0$ in $\ltwo(\mu^M)$, we can extract a
subsequence of $\{\lambda^{n_k}\}_k$ which converges $\mu^M$-almost
everywhere to $\lambda^0$. We denote this subsequence by
$\{\lambda^{n_k}\}_k$, as well. By Fatou's lemma (applied to the
$d\ab{M}$-integrals) we have
\begin{align}  \label{Fatou1}
\liminf_{k \to \infty} \int_0^T (\lambda^{n_k}_u)^2\, d\ab{M}_u \ge
\int_0^T (\lambda^0_u)^2\, d\ab{M}_u.
\end{align}
Another application of Fatou's lemma (this time with respect to
the probability $\PP$) and the fact that $||
\lambda^{n_k}||^2_{\ltwo(\mu^M)}\to
||\lambda^0||^2_{\ltwo(\mu^M)}$ imply that
\begin{align*}
\mathbb{E}\left[\int_0^T (\lambda_u^0)^2\, d\ab{M}_u\right]
&=\lim_{k\to\infty} \mathbb{E}\left[\int_0^T (\lambda^{n_k}_u)^2\,
d\ab{M}_u\right]\\ &\ge\mathbb{E}\left[\liminf_{k\to \infty} \int_0^T
(\lambda^{n_k}_u)^2\, d\ab{M}_u \right]\ge \mathbb{E}\left[%
\int_0^T (\lambda_u^0)^2\, d\ab{M}_u\right]
\end{align*}
which shows that we have equality in \eqref{Fatou1}, $\PP$-almost
surely. To extract an a.s.-convergent subsequence from $\int_0^T
(\lambda^{n_k}_u)^2\, d\ab{M}_u$ - and finish the proof - all we need to do
is apply the result of Lemma \ref{inverse-Fatou} below.
\end{proof}

\begin{lemma}
\label{inverse-Fatou} Any sequence $\{f^k\}_{k\in\N}\subseteq\lone(\PP)$ of
non-negative random variables  which satisfies the two properties
\begin{align}\label{Fatou2}
\lim_{n\to\infty}\EE[f^k] = \EE[f^0],\quad \liminf_{k\to\infty}
f^k = f^0 \quad \PP-a.s.
\end{align}
for some $f^0\in \lone(\PP)$, has a subsequence $\{f^{k_l}\}_{l\in\N}$
converging almost surely to $f^0$.
\end{lemma}
\begin{proof}
From \eqref{Fatou2} and the Lebesgue's
theorem of monotone convergence we have
\begin{equation*}
\lim_{k\to\infty} \mathbb{E}\left[\inf_{m\ge k}f^m\right] =
\mathbb{E}[f^0] = \lim_{k \to\infty}\mathbb{E}[f^k]
\end{equation*}
which means that $f^k-\inf_{m\ge k} f^m \to 0$ in $\mathbb{L}^1$. We
can, therefore, extract a subsequence $\{f^{k_l}\}_{l\in\N}$ of
$\{f^k\}_{k\in\N}$ such that $f^{k_l}-\inf_{m\geq l} f^{k_m}$
converges to $0$ $\PP$-a.s. Thanks to monotonicity of $\inf_{m\geq k}
f^m$, the sequence $f^{k_l}$ must itself converge $\PP$-a.s.~towards
$\lim_{k\to\infty} \inf_{m\geq k} f^m= \liminf_{k} f^k=f^0$.
\end{proof}

\noindent{\bf Acknowledgements:} The authors would like to thank
Morten Mosegaard Christensen, Philip Protter, the participants at the
third Carnegie Mellon-Columbia-Cornell-Princeton Conference and an
anonymous referee for fruitful discussions.

\def\cprime{$'$} \def\cprime{$'$}
\providecommand{\bysame}{\leavevmode\hbox to3em{\hrulefill}\thinspace}
\providecommand{\MR}{\relax\ifhmode\unskip\space\fi MR }
\providecommand{\MRhref}[2]{%
  \href{http://www.ams.org/mathscinet-getitem?mr=#1}{#2}
}
\providecommand{\href}[2]{#2}


\begin{thebibliography}{EKJPS98}

\bibitem[AS92]{AnsStr92}
Jean-Pascal Ansel and Christophe Stricker, \emph{Lois de martingale, densit\'es
  et d\'ecomposition de {F}\"ollmer-{S}chweizer}, Ann. Inst. H. Poincar\'e
  Probab. Statist. \textbf{28} (1992), no.~3, 375--392.

\bibitem[AS93]{AnsStr93}
J.~P. Ansel and C.~Stricker, \emph{Unicit\'e et existence de la loi minimale},
  S\'eminaire de Probabilit\'es, XXVII, Lecture Notes in Math., vol. 1557,
  Springer, Berlin, 1993, pp.~22--29.

\bibitem[CH89]{CoxHua89}
J.~C. Cox and C.~F. Huang, \emph{Optimal consumption and portfolio policies
  when asset prices follow a diffusion process}, J. Economic Theory \textbf{49}
  (1989), 33--83.

\bibitem[CR05]{CarRas05}
Laurence Carasus and Mikl{\' o}s R{\' a}sonyi, \emph{Optimal strategies and
  utility-based prices converge when agents' preferences do}, preprint, 2005.

\bibitem[DS94]{DelSch94}
Freddy Delbaen and Walter Schachermayer, \emph{A general version of the
  fundamental theorem of asset pricing}, Math. Ann. \textbf{300} (1994), no.~3,
  463--520.

\bibitem[DS95]{DelSch95}
F.~Delbaen and W.~Schachermayer, \emph{The existence of absolutely continuous
  local martingale measures}, Ann. Appl. Probab. \textbf{5} (1995), no.~4,
  926--945.

\bibitem[DS98]{DelSch98a}
Freddy Delbaen and Walter Schachermayer, \emph{A simple counterexample to
  several problems in the theory of asset pricing}, Math. Finance \textbf{8}
  (1998), no.~1, 1--11.

\bibitem[EKJPS98]{ElkJeaShr98}
Nicole El~Karoui, Monique Jeanblanc-Picqu{\'e}, and Steven~E. Shreve,
  \emph{Robustness of the {B}lack and {S}choles formula}, Math. Finance
  \textbf{8} (1998), no.~2, 93--126.

\bibitem[FS02]{FolSch02a}
Hans F{\"o}llmer and Alexander Schied, \emph{Stochastic finance}, de Gruyter
  Studies in Mathematics, vol.~27, Walter de Gruyter \& Co., Berlin, 2002, An
  introduction in discrete time.

\bibitem[GS89]{GilSch89}
Itzhak Gilboa and David Schmeidler, \emph{Maxmin expected utility with
  nonunique prior}, J. Math. Econom. \textbf{18} (1989), no.~2, 141--153.

\bibitem[Had02]{Had02}
Jacques Hadamard, \emph{Sur les probl\`emes aux d\'eriv\'ees partielles et leur
  signification physique}, Princeton University Bulletin (1902), 49--52.

\bibitem[HP91]{HePea91}
H.~He and N.~D. Pearson, \emph{Consumption and portfolio policies with
  incomplete markets and short-sale constraints: the finite-dimensional case},
  Mathematical Finance \textbf{1} (1991), 1--10.

\bibitem[HS98]{HubSch98}
Friedrich Hubalek and Walter Schachermayer, \emph{When does convergence of
  asset price processes imply convergence of option prices?}, Math. Finance
  \textbf{8} (1998), no.~4, 385--403.

\bibitem[JN04]{JouNap04}
Ely{\`e}s Jouini and Clotilde Napp, \emph{Convergence of utility functions and
  convergence of optimal strategies}, Finance Stoch. \textbf{8} (2004), no.~1,
  133--144.

\bibitem[JS03]{JacShi03}
Jean Jacod and Albert~N. Shiryaev, \emph{Limit theorems for stochastic
  processes}, second ed., Grundlehren der Mathematischen Wissenschaften
  [Fundamental Principles of Mathematical Sciences], vol. 288, Springer-Verlag,
  Berlin, 2003.

\bibitem[KLSX91]{KarLehShrXu91}
I.~Karatzas, J.~P. Lehoczky, S.~E. Shreve, and G.~L. Xu, \emph{Martingale and
  duality methods for utility maximization in an incomplete market}, SIAM
  Journal of Control and Optimisation \textbf{29(3)} (1991), 702--730.

\bibitem[KS99]{KraSch99}
D.~Kramkov and W.~Schachermayer, \emph{The asymptotic elasticity of utility
  functions and optimal investment in incomplete markets}, Ann. Appl. Probab.
  \textbf{9} (1999), no.~3, 904--950.

\bibitem[K{\v Z}03]{KarZit03}
I~Karatzas and G.~{\v Z}itkovi{\' c}, \emph{Optimal consumption from investment
  and random endowment in incomplete semimartingale markets}, Ann. Probab.
  \textbf{31} (2003), no.~4, 1821--1858.

\bibitem[MMR04]{MMR04} Maccheroni, F.,
Marinacci, M. \& Rustichini, A. \emph{Variational representation
of preferences under ambiguity}. Working paper no.5, ICER working
paper series. (2004).

\bibitem[Mer71]{Mer71}
R.~C. Merton, \emph{Optimum consumption and portfolio rules in a
  conitinuous-time model}, J. Economic Theory \textbf{3} (1971), 373--413.

\bibitem[Pli86]{Pli86}
S.~R. Pliska, \emph{A stochastic calculus model of continuous trading: optimal
  portfolio}, Math. Oper. Res. \textbf{11} (1986), 371--382.

\bibitem[Pri03]{Pri03}
Jean-Luc Prigent, \emph{Weak convergence of financial markets}, Springer
  Finance, Springer-Verlag, Berlin, 2003.

\bibitem[Pro04]{Pro04}
Philip~E. Protter, \emph{Stochastic integration and differential equations},
  second ed., Applications of Mathematics (New York), vol.~21, Springer-Verlag,
  Berlin, 2004, Stochastic Modelling and Applied Probability.

\bibitem[Roc70]{Roc70}
R.~T. Rockafellar, \emph{Convex {A}nalysis}, Princeton University Press,
  Princeton, 1970.

\bibitem[Rog01]{Rog01}
L.~C.~G. Rogers, \emph{The relaxed investor and parameter uncertainty}, Finance
  Stoch. \textbf{5} (2001), no.~2, 131--154.

\bibitem[Sch95]{Sch95a}
Martin Schweizer, \emph{On the minimal martingale measure and the
  {F}\"ollmer-{S}chweizer decomposition}, Stochastic Anal. Appl. \textbf{13}
  (1995), no.~5, 573--599.

\bibitem[TV02]{Troj2002} F. Trojani and P. Vanini, \emph{A note on robustness in Merton's model of intertemporal
consumption and portfolio choice}, Journal of economic dynamics
and control. {\bf26} (2002), 423--435.

\end{thebibliography}
\end{document}